\documentclass[conference]{IEEEtran}
\ifCLASSINFOpdf
\else
\fi
\usepackage{url}

\usepackage{hyperref}

\usepackage{amsmath}
\usepackage{amsfonts}
\usepackage{amssymb}
\usepackage{amsthm}

\usepackage{enumerate}

\newtheorem{definition}{Definition}
\newtheorem*{itheorem}{Theorem}
\newtheorem{theorem}{Theorem}
\newtheorem{lemma}{Lemma}
\newtheorem{corollary}{Corollary}
\newtheorem{claim}{Claim}
\newtheorem{remark}{Remark}

\usepackage[backend=bibtex,style=alphabetic,maxnames=5,minnames=5,maxbibnames=99,firstinits=true,
doi=false,isbn=false,url=false
]{biblatex}
\AtEveryBibitem{%
    \clearfield{pages}
}

\renewcommand\IEEEkeywordsname{Keywords}

\usepackage[nameinlink]{cleveref}

\Crefname{theorem}{Theorem}{Theorems}
\Crefname{section}{Section}{Sections}

\begin{document}
%
\title{Strong Chain Rules for Min-Entropy under Few Bits Spoiled}

\author{
\IEEEcompsocitemizethanks{\IEEEcompsocthanksitem supported by}
\IEEEauthorblockN{Maciej Sk\'{o}rski}
\IEEEauthorblockA{IST Austria\\
Email: mskorski@ist.ac.at}\\
}\thanks{Supported by}


%


\maketitle

\begin{abstract}
It is well established that the notion of min-entropy fails to satisfy the \emph{chain rule} of the form 
$H(X,Y) = H(X|Y)+H(Y)$, known for Shannon Entropy. 
The lack of a chain rule causes a lot of technical difficulties, particularly in cryptography 
where the chain rule would be a natural way to analyze how min-entropy is split among smaller blocks. Such problems arise
for example when constructing extractors and dispersers.

We show that any sequence of variables exhibits
a very strong strong block-source structure (conditional distributions of blocks are nearly flat) when we \emph{spoil few correlated bits}.
This implies, conditioned on the spoiled bits, that \emph{splitting-recombination properties} hold.
In particular, we have many nice properties that min-entropy doesn't obey in general, for example
strong chain rules, ``information can't hurt'' inequalities, equivalences of average and worst-case conditional entropy definitions 
and others. 

Quantitatively, for any sequence $X_1,\ldots,X_t$ of random variables over an alphabet $\mathcal{X}$
we prove that, when conditioned on $m = t\cdot O( \log\log|\mathcal{X}| + \log\log(1/\epsilon) + \log t)$ bits of auxiliary information,
all conditional distributions of the form $X_i|X_{<i}$ are $\epsilon$-close to be nearly flat (only a constant factor away).
The argument is combinatorial (based on simplex coverings).

This result may be used as a generic tool for \emph{exhibiting block-source structures}.
We demonstrate this by reproving the fundamental converter due to Nisan and Zuckermann (\emph{J. Computer and System Sciences, 1996}), 
which shows that sampling blocks from a min-entropy source roughly preserves the entropy rate.
Our bound implies, only by straightforward chain rules, an additive loss of $o(1)$ (for sufficiently many samples), which qualitatively meets 
the first tighter analysis of this problem due to Vadhan (\emph{CRYPTO'03}), obtained by large deviation techniques.
\end{abstract}
\vspace{0.5cm}


%
\IEEEpeerreviewmaketitle

\section{Introduction}

\subsection{Strong vs Weak Entropy Chain Rules}
One of the most useful properties of Shannon entropy is the chain rule, 
showing how entropy splits between distributions
\begin{align}\label{eq:chain_rule}
 H_1(X|Y) = H(X,Y) - H(Y).
\end{align}
The notion of min-entropy, very important for cryptography~\cite{DBLP:conf/icalp/Shaltiel11}, fails to satisfy this property~\cite{DBLP:conf/icits/IwamotoS13}. 
In the lack of a chain rule, a much weaker one-sided bound (e.g. of the form
$ \widetilde{H}_{\infty}(X|Y) \geqslant H_{\infty}(X)-H_{0}(Y)$, where $\tilde{H}_{\infty}$ is an appropriate extension of min-entropy to conditional distributions)
is sometimes used~\cite{DBLP:journals/siamcomp/DodisORS08}, which we address as a \emph{weak chain rule}\footnote{In leakage-resilient cryptography such bounds are simply called chain rules.
In this paper we discuss chain rules in a strong sense.}.

\subsection{Need for Strong Chain Rules}

While the weak chain rule suffices for many applications related to bounded leakage
~\cite{DBLP:journals/siamcomp/DodisORS08,DBLP:conf/focs/DziembowskiP08} (where $Y$ is leakage much shorter than the amount of min-entropy in $X$),
it is insufficient where one needs to estimate how entropy is distributed among blocks. In these settings,
one would like to argue that (roughly) either $X$ or $Y|X$ has high min-entropy if the joint min-entropy (of $(X,Y)$) is high. 
Examples of such problems are randomness extraction in the bounded storage model~\cite{DBLP:conf/eurocrypt/Cachin97}, constructions of dispersers~\cite{Barak:2006:DSE:1132516.1132611}, 
or oblivious transfer protocols~\cite{DBLP:conf/crypto/DamgardFRSS07}.





\subsection{Our Contribution and Related Works}

Although the chain rule fails in general, we show that it is true conditioned on \emph{few spoiled bits}.
We actually show more, that (locally, conditioned on auxiliary bits) a very strong \emph{block-source structure} exists. Namely, each block
is nearly flat given previous blocks. Informally, the theorem reads as follows
\begin{itheorem}[Informal: exhibiting flat block-source structures]
For any sequence $X=X_1,\ldots,X_t$ of correlated random variables each over $\mathcal{X}$ and any $\epsilon$
there exists auxiliary information $S$ of length $m$ bits such that
\begin{enumerate}[(a)]
 \item $S$ is short: $m=t\cdot O(\log\log|\mathcal{X}|+\log\log(1/\epsilon) + \log t)$
 \item Conditioned on $S$, conditional block distributions 
$\Pr[X_i|X_{i}=x_{i},X_{i-1}=x_{i-1},\ldots,X_1=x_1]$ are nearly flat ($\epsilon$-close to
a probability distribution whose values differ by a constant factor).
\end{enumerate}
\end{itheorem}
The formal statement is given in \Cref{thm:chain_rule}.
For cryptographic applications, $\log\log(\epsilon^{-1})$ is pretty much a small constant (typically $\epsilon = 2^{-100}$). 
Also, for sources with super-logarithmic entropy per block, that is when $H_{\infty}(X) \gg t \log \log|\mathcal{X}|)$,
and the number of blocks $t$ growing not too fast, e.g. $t = \log^{O(1)}(|\mathcal{X}|)$, 
the error term is of a smaller order than the entropy. Under these mild assumptions, conditioned on the partition generated
by auxiliary bits, we conclude
many nice properties that fail in general. Examples are chain rules, ``conditioning only decreases entropy'' properties,
equivalences of conditional entropy defined in different ways and others.


\subsubsection{Our Tools}

\paragraph{Spoiling Knowledge}
The \emph{spoiling knowledge} technique is essentially about finding auxiliary information that increases entropy, and
was introduced in \cite{DBLP:journals/tit/BennettBCM95}. We use the same idea to force block distributions to be nearly flat.

\paragraph{Covering techniques (combinatorial geometry)}
In order to construct a good ``spoiling'', we consider the logarithm of the chaining identity 
$p_{X_1\ldots X_t}(\cdot) = \prod_{j=1}^{t}p_{X_i|X_{i-1}\ldots X_1}(\cdot)$ which 
represents the ``surprise'' of the total distribution as a sum of ``next-block surprises''
of the form $r_i(\cdot) = -\log p_{X_i|X_{i-1}\ldots X_1}(\cdot)$. It follows
that the vector of all $r_i$ (for $i=1,\ldots,t$) lie in a $(t-1)$-dimensional simplex
of edge roughly $O(t\log|\mathcal{X}|)$. Our partition is obtained from \emph{coverings}, as 
for all $x$ in the same part $r_i(\cdot)$ is roughly constant (when the radius is small enough).
This approach can be likely optimized (we use a crude bound on the covering number). Also the lower bounds on the necessary number of spoiled buts are possible, by considering packings instead of coverings.
We defer this discussion to the full version.



\subsection{Related works}
As far as we know, the presented result on spoiling min-entropy is knew.
The Nisan-Zuckerman lemma, discussed in this paper as an application,
was analyzed by Vadhan~\cite{DBLP:conf/crypto/Vadhan03} and recently by Bellare~\cite{DBLP:conf/crypto/BellareKR16}. 
These results study min-entropy present in random chunks of a larger source,
and don't offer tools for \emph{splitting entropy} in any source \emph{deterministically} (particularly for a small number of blocks),
as we do.

\subsection{Applications}

The important result due to Nisan and Zuckerman~\cite{DBLP:journals/jcss/NisanZ96}, improved later by Vadhan~\cite{DBLP:conf/crypto/Vadhan03} 
states that sampling from a given source of high min-entropy rate $\alpha$ yields
a source of a comparable entropy rate $\beta$. This fact is a crucial step in constructions
of so called \emph{local extractors}, that extract randomness parsing only a part of input. 
For a while, only a lossy bound $\beta \approx \alpha/\log(1/\alpha)$ was known.
The reason was precisely the lack of a chain rule for min-entropy.
As observed by Cachin~\cite{DBLP:conf/eurocrypt/Cachin97} the proof for Shannon entropy (a less interesting case) is straightforward and follows basically by a splitting-recombining argument, which uses a chain rule in both directions.
We demonstrate by our technique that (surprisingly) a very effective splitting-recombining approach actually works,
and achieves $\beta = \alpha- o(1)$ in a very straightforward way. This matches the bound due to Vadhan.
Concretely, if the original source is a sequence of $t$ blocks over an alphabet $\mathcal{X}$ and we take $\ell$ samples, then
$\beta = \alpha - err_{\mathsf{Spoil}}  -err_{\mathsf{Samp}}$
where the losses due to chain rules and sampling equal, respectively 
\begin{align*}
err_{\mathsf{Spoil}} &= O(\log\log|\mathcal{X}|+\log\log(1/\epsilon)+\log t)/\log|\mathcal{X}| \\ 
err_{\mathsf{Samp}} &= O(\sqrt{\ell^{-1}\log(1/\epsilon)}).
\end{align*}
In particular $\beta$ converges to $\alpha$ when the block length is $\log|\mathcal{X}| = \log^{\omega(1)}(1/\epsilon)$ and 
$\log |\mathcal{X}| = \omega(1)\cdot t $.
For more details see \Cref{lemma:sampling} in \Cref{sec:app}.
For the discussed result our bounds converge slightly slower than Vadhan's bounds derived by large deviation techniques.
However, our spoiling technique can be used also for small number of samples.

\subsection{Organization}

In \Cref{sec:prelim} we explain necessary notions and notations. Auxiliary facts that will be needed are discussed
in \Cref{sec:lemmas}. In \Cref{sec:main} we prove the main result.
Applications to the bounded storage model are discussed in \Cref{sec:app}. 
We conclude the work in \Cref{sec:conclusion}.

\section{Preliminaries}\label{sec:prelim}

\subsection{Basic Notation}
For any random variables $X_1,X_2$ by $p_{X_1|X_2}$ we denote the distribution of $X_1$ conditioned on $X_2$, that is
 $p_{X_1|X_2}(x_1,x_2) = \Pr[X_1=x_1|X_2=x_2]$. 
 Throughout this paper, all logarithms are taken to base $2$.
 For any sequence of random variables $X = X_1,\ldots,X_n$
we denote $X_{<i}=X_{1},\ldots,X_{i-1}$,  $X_{\leqslant i}=X_{1},\ldots,X_{i}$
and more generally for any subset $I\subset \{1,\ldots,n\}$ we put $X_{I} = X_{i_1}X_{i_2}\ldots X_{i_{m-1}}X_{i_{m}}$ where $i_1<i_2<\ldots <i_m$ are all elements of $I$.

\subsection{Distances, Entropies}
In the definitions below $\mathcal{X}$ is an arbitrary finite set.
\begin{definition}[Statistical Distance]
For two random variables $X,Y$ on $\mathcal{X}$ by
the statistical distance (total variation) we mean
\begin{align*}
 d_{TV}(X;Y) = \frac{1}{2}\sum_{x\in\mathcal{X}}| \Pr[X=x] - \Pr[Y=x]|
\end{align*}
\end{definition}


\begin{definition}[Shannon Entropy]
The Shannon entropy of a random variable $X$ on $\mathcal{X}$ equals
\begin{align*}
 H_{1}(X) = -\sum_{x}\Pr[X=x]\log  \Pr[X=x].
\end{align*}
\end{definition}

\begin{definition}[Min-Entropy]
The min-entropy of a random variable $X$ on $\mathcal{X}$ equals
\begin{align*}
 H_{\infty}(X) = -\log \max_{x\in\mathcal{X}} \Pr[X=x].
\end{align*}
\end{definition}

\begin{definition}[Conditional Min-Entropy~\cite{DBLP:journals/siamcomp/DodisORS08}]
Let $X,Y$ be random variables over $\mathcal{X}$ and $\mathcal{Y}$ respectively.
The \emph{worst-case} min-entropy of $X$ conditioned on $Y$ equals
\begin{align*}
 {H}_{\infty}(X|Y) = \min_{y\in\mathcal{Y}} H_{\infty}(X|Y=y).
\end{align*}
The \emph{average} min-entropy of $X$ conditioned on $Y$ equals
\begin{align*}
 \widetilde{H}_{\infty}(X|Y) = -\log \left(\mathbb{E}_{y\gets Y} 2^{-H_{\infty}(X|Y=y)} \right).
\end{align*}
\end{definition}
\begin{remark}
The averaged notion is slightly weaker, but has better properties and actually better suits applications~\cite{DBLP:journals/siamcomp/DodisORS08}.
\end{remark}
The notion of smooth entropy is more accurate than min-entropy because quantifies entropy up to small perturbations in the probability mass.
\begin{definition}[Smooth Min-Entropy~\cite{WolfRenner2004}]
The \emph{$\epsilon$-smooth min-entropy} of a random variable $X$ on $\mathcal{X}$ is defined as
\begin{align*}
  H_{\infty}^{\epsilon}(X)  = \max_{X': d_{TV}(X,X')\leqslant \epsilon} H_{\infty}(X')
\end{align*}
where the maximum is over all random variables $X'$ on $\mathcal{X}$. In other words, $X$ has at least $k$ bits of \emph{smooth min-entropy}
if there is $X'$ of min-entropy at least $k$ and $\epsilon$-close to $X$.
\end{definition}


\subsection{Randomness Extractors}

Below we recall the definition of seeded extractors
\begin{definition}[Randomness Extractor~\cite{DBLP:journals/jcss/NisanZ96}]
We say that a function $\mathsf{Ext}:\{0,1\}^n\times\{0,1\}^d \rightarrow \{0,1\}^m$
is a $(k,\epsilon)$-extractor if and only if
\begin{align*}
 d_{TV}(\mathsf{Ext}(X,U_d),U_d ; U_m,U_d) \leqslant \epsilon
\end{align*}
for any $X$ on $\{0,1\}^n$ with min-entropy at least $k$.
\end{definition}

\subsection{Samplers}

Averaging samplers are procedures which sample points, within a given domain, that are distributed enough random to 
approximate every function. It turns out that there exist good averaging samples, using
much less auxiliary randomness than necessary to generate independent points.
For our applications we consider samplers that approximate averages from below.
\begin{definition}[Averaging Samplers~\cite{DBLP:conf/crypto/Vadhan03}]\label{def:samplers}
A function $\mathsf{Samp}:\{0,1\}^d \rightarrow [t]^{\ell}$ is a $(\mu,\theta,\gamma)$-averaging sampler if for every function $f:[t]\rightarrow [0,1]$
with average value $\frac{1}{\ell}\sum_{i=1}^{\ell}f(i) \geqslant \mu$ it holds that
\begin{align*}
 \Pr_{i_{1},\ldots,i_{\ell}\gets \mathsf{Samp}(U_r)}\left[ \frac{1}{\ell}\sum_{i=1}^{\ell}f(i) \geqslant \mu - \theta \right] \leqslant 1-\gamma.
\end{align*}
\end{definition}

\begin{lemma}[Optimal Averaging Samplers (Nonconstructive)~\cite{DBLP:conf/crypto/Vadhan03}]
For the setting in \Cref{def:samplers}
there is an averaging sampler which for any $\ell'$ such that
$\Omega(\mu\theta^{-2}\log(1/\gamma) ) \leqslant \ell' \leqslant \ell$
produces $\ell'$ distinct samples and uses
$d = \log(t/\ell') + \log(1/\gamma) + 2\log(\mu/\theta) + \log\log(1/\mu)+O(1)$ random bits.
\end{lemma}

\section{Auxiliary Lemmas}\label{sec:lemmas}

The following lemma is essentially the \emph{information can't hurt} principle, well known for Shannon entropy, stated for the notion of min-entropy. 
It says that extra information can only reduce the uncertainty, and comes up often in cryptography.
\begin{lemma}[Conditioning only reduces entropy~\cite{DBLP:conf/icits/IwamotoS13}]\label{lemma:conditioning_reduces_entropy}
For any $X,Y,Z$ we have $\widetilde{H}_{\infty}(X|Y,Z) \leqslant \widetilde{H}_{\infty}(X|Y)$.
\end{lemma}

It is well known (see, e.g., \cite{DBLP:conf/ches/BarakST03}) that when each block has certain min-entropy entropy conditioned on previous blocks, the total entropy
grows linearly with the number of blocks (as expected). The lemma is stated without a proof (it follows directly from definitions).
We stress, however, that using the \emph{worst-case} notion of conditional entropy is crucial. 
\begin{lemma}[Min-entropy from block sources]\label{lemma:chaining_entropy}
For any $X,Y,Z$ 
we have 
${H}_{\infty}(X,Y|Z) \geqslant {H}_{\infty}(X|Y,Z)+{H}_{\infty}(Y|Z)$.
\end{lemma}

Another lemma, well known in folklore, quantifies the intuition that conditioning on ``large'' events (not too surprising information) doesn't decrease entropy too much.
\begin{lemma}[Leakage lemma for min-entropy]\label{lemma:leakage}
For any random variable $X$ and any events $A,B$ we have $H_{\infty}(X|B,A) \geqslant H_{\infty}(X|B)-\log(1/\Pr(A|B))$. 
\end{lemma}
\begin{proof}
The first part of the lemma follows because we have $\Pr[X=x|A,B] = \Pr[X=x,A\cap B]/\Pr[A\cap B] \leqslant \Pr[X=x,B]/(\Pr[B]\cdot \Pr[A|B])$ for every $x$.
\end{proof}



\section{Main Result}\label{sec:main}

\subsection{Chain Rule}


\begin{theorem}[Exhibiting block structures with few bits spoiled]\label{thm:chain_rule}
Let $\mathcal{X}$ be a fixed alphabet and $X=(X_1,\ldots,X_t)$ be a sequence of (possibly correlated) random variables each over $\mathcal{X}$.
Then for any $1>\epsilon>0$ and $\delta>0$ there exists a collection $\mathcal{B}$ of disjoint sets on $\mathcal{X}^t$
such that 
\begin{enumerate}[(a)]
\item $\mathcal{B}$ can be indexed by a small number of bits, namely
\begin{align*}
\log |\mathcal{B}| = t\cdot O\left(\log\log|\mathcal{X}| +\log \log(\epsilon^{-1})+\log (t/\delta)\right)
\end{align*}
\item $\mathcal{B}$ almost covers the domain
\begin{align*}
 \sum_{B\in \mathcal{B}}p_X(B) \geqslant 1-\epsilon &
\end{align*}
\item Conditioned on members of $\mathcal{B}$, block distributions $X_i|X_{<i}$ are nearly flat.
\begin{align*}
\forall x,x'\in B:\quad 2^{-O(\delta)} \leqslant \frac{p_X(x_i|x_{<i})}{p_X(x'_i|x'_{<i})} \leqslant 2^{O(\delta)}.
\end{align*}
for every $B\in\mathcal{B}$ and $i=1,\ldots,t$.
\end{enumerate}
\end{theorem}
In some applications (see for example \Cref{sec:app}) it is convenient to work with parts that are not too small. 
By substituting $\epsilon:=\frac{\epsilon}{2}$, and deleting from $\mathcal{B}$
all members of smaller than $\frac{\epsilon}{2|\mathcal{B}|}$ (their mass is at most $\frac{\epsilon}{2}$), we obtain
\begin{remark}[Getting rid of tiny parts]\label{rem:fat_partition}
In \Cref{thm:chain_rule}, we may assume $p_X(B) = \Omega\left(\epsilon\cdot |\mathcal{B}|^{-1}\right)$ for every $B\in\mathcal{B}$.
\end{remark}

\begin{corollary}[Conditional entropies under few bits spoiled]\label{cor:chain_rule}
Under the assumptions of \Cref{thm:chain_rule}, for every $B\in\mathcal{B}$ for every index $i$ and for every 
set $I\subset \{1,\ldots,i-1\}$ we have
\begin{enumerate}[(a)]
 \item The chain rule for min-entropy
 \begin{align*}
 H_{\infty}(X_i|X_I,B) = H_{\infty}(X_i,X_I|B) - H_{\infty}(X_I|B)  \pm O(\delta).
\end{align*}
\item The average and worst-case min-entropy almost match
\begin{align*}
\widetilde{H}_{\infty}(X_i|X_I,B) = {H}_{\infty}(X_i|X_I,B)  \pm O(\delta).
\end{align*}
\end{enumerate}
\end{corollary}

\begin{proof}[Proof of \Cref{cor:chain_rule}]
Fix any subset $I\subset \{1,\ldots,i-1\}$ of size $m$. By \Cref{thm:chain_rule} for every $x$ and $B$ we have
\begin{align}\label{cor:1}
 \Pr[X_{\leqslant i} = x_{\leqslant i}  |B ] =  2^{\pm O(\delta)}\Pr[X_{< i} = x_{< i} |B ].
\end{align}
Let $J = \{1,\ldots,i-1\}\setminus I$. Taking the sum of \Cref{cor:1} over
$x$ such that $x_{I} = x'_{I}$ and $x_i = x''$ are fixed but $x_{J}$ varies we obtain
\begin{align*}
 \Pr[X_i=x'',X_{I} = x'_{I}  |B ] =  2^{\pm O(\delta)}\Pr[X_{I} = x'_{I} |B ].
\end{align*}
This implies
\begin{align*}
H_{\infty}(X_i|X_I,B) = H_{\infty}(X_i,X_I,B) - H_{\infty}(X_I,B)  \pm O(\delta) \\
\widetilde{H}_{\infty}(X_i|X_I,B) = {H}_{\infty}(X_i|X_I,B)   \pm O(\delta) 
\end{align*}
which finishes the proof.
\end{proof}


\begin{proof}[Proof of \Cref{thm:chain_rule}]
Let $p$ be the joint distribution of $X_1,\ldots,X_n$.
For any $i$ denote 
\begin{align}\label{eq:decomposition}
p_i(x_i|x_{i-1},\ldots,x_1) = p_{X_i|X_{<i}}(x_i,\ldots,x_1).
\end{align}
and let the ``surprise'' of the bit $x_i$ be
\begin{align*}
 r^i(x) = -\log p_i(x_i,\ldots,x_1).
\end{align*}
Note that
$p(x_1,\ldots,x_n) = \prod_{i=1}^{n} p_i(x_i,x_{i-1},\ldots,x_1) 
$ and therefore, denoting $x=(x_1,\ldots,x_n)$, we obtain
\begin{claim}[Decomposing surprises]
We have
\begin{align}\label{eq:cut}
 \sum_{i=1}^{t}r^i(x) = -\log p(x).
\end{align}
\end{claim}
The next claim follows by a simple Markov-type argument.
\begin{claim}[Significant probabilities]\label{claim:significant}
There exists a set $A\subset \mathcal{X}^t$ of
probability $1-\epsilon$ such that $p(x) \geqslant \frac{\epsilon}{|\mathcal{X}|^t}$ 
for all $x\in A$.
\end{claim}

Denoting $x=x_1,\ldots,x_n$, we have
\begin{align}
\forall x\in A:\quad -\log p(x_1,\ldots,x_n) \leqslant t\log|\mathcal{X}| + \log(1/\epsilon).
\end{align}
The claim below follows from \Cref{eq:decomposition} and \Cref{eq:cut}
\begin{claim}[Surprises live in a simplex]
We have $r^i(x) \geqslant 0$ for $i=1,\ldots, t$ and $\sum_{i=1^{t}} r^i(x) \leqslant t\log|\mathcal{X}| + \log(1/\epsilon)$,
for all points $x\in A$.
\end{claim}

\begin{claim}[Simplex coverings imply a chain rule]\label{clam:coverings_to_chainrule}
If the simplex with side length $t\log|\mathcal{X}| + \log(1/\epsilon)$ can be covered 
by $N$ balls of radius $R$ in the $\ell_{\infty}$ norm, then the theorem holds with $|\mathcal{B}| =  N$ and $\delta =  R$.
\end{claim}
\begin{proof}[Proof of \Cref{clam:coverings_to_chainrule}]
Let $C\subset \mathbb{R}^{t}$, $|C|=N$, be the set of the centers of the covering balls.
Let $S$ be the function which assigns to every point $x\in A$ (where $A$ is defined in \Cref{claim:significant}) the point $z\in C$ 
closest to the vector $(r_1(x),\ldots,r_t(x))$ in the $\ell_{\infty}$-norm.
Fix any $z$ and let $B_z = \{x: S(x)=z\}$. By the properties of the covering, for $i=1,\ldots,t$ we obtain
\begin{align}\label{eq:surprises_concentrated}
\forall x\in B: \quad \left|r^{i}(x)-z_i\right| \leqslant R
\end{align}
In particular, the surprises for any two points in $B$ are close
\begin{align}\label{eq:surprises_concentrated2}
 \forall x,x'\in B: \quad \left|r^{i}(x) -r^{i}(x')\right| \leqslant 2R, \quad i=1,\ldots,t.
\end{align}
Let $p_B$ be the conditional probability of $p$ given $B$. Denote by 
$r^{i}_B$  the surprise of the $i$-th bit given previous bits and conditioned on $B$, that is
\begin{align*}
r^{i}_B(x) &= -\log\Pr[X_i=x_i|X_{i-1}=x_{i-1},\ldots,X_1=x_1,B] \\
& = -\log p_B(x_i | x_{i-1},\ldots,x_1)
\end{align*}
Note that $p_B(x) = p(x)/p(B)$ for $x\in B$, and hence
$r^{i}_B(x) = r^{i}(x) + \log (p(B)^{-1})$.
Now \Cref{eq:surprises_concentrated2} implies
\begin{align}\label{eq:surprises_concentrated2_cond}
 \forall x,x'\in B: \quad \left|r^{i}_B(x) -r^{i}_B(x')\right| \leqslant 2R, \quad i=1,\ldots,t.
\end{align}
which finishes the proof.
\end{proof}
It remains to observe that the covering number for our case is
$\log N = \log N_0 + \log\log(|\mathcal{X}|)+\log\log(1/\epsilon))+\log (t/\delta)$ (see for example \cite{simplex})
which finishes the proof.
\end{proof}

\section{Applications}\label{sec:app}

\begin{theorem}[Sampling preserves entropy rate~\cite{DBLP:conf/crypto/Vadhan03}]\label{lemma:sampling}
Let $\mathcal{X}$ be a fixed finite alphabet, and
let $X_1,\ldots,X_t$ be a sequence of correlated random variables each over $\mathcal{X}$. 
Let $i_1,\ldots,i_{\ell} \in [1,t]$, where $\ell < t$, be chosen from the set $\{1,\ldots,t\}$ by an averaging $(\mu,\theta,\gamma)$-sampler. Then there is a random variable $\mathcal{B}$ 
taking $m = t\cdot O(\log\log|\mathcal{X}|+\log\log(1/\epsilon)+\log t)$ bits, 
such that
\begin{multline*}
\frac{1}{\ell\log|\mathcal{X}|}\widetilde{H}^{\epsilon}_{\infty}(X_{i_{\ell}}X_{i_{\ell-1}}X_{i_{\ell-2}}\ldots,X_{i_1}|\mathcal{B}) \geqslant \\
\frac{1}{t\log|\mathcal{X}|}\widetilde{H}_{\infty}(X_{i_{\ell}}X_{i_{\ell-1}}X_{i_{\ell-2}}\ldots,X_{i_1})-err_{\mathsf{Spoil}}-err_{\mathsf{Sampler}}
\end{multline*}
where the errors due to spoiling and sampling equal
\begin{align*}
err_{\mathsf{Spoil}} &= O(\log\log|\mathcal{X}|+\log\log(1/\epsilon)+\log t)/\log|\mathcal{X}| \\
err_{\mathsf{Samp}} &= O(\sqrt{\ell^{-1}\log(1/\epsilon)}).
\end{align*}
\end{theorem}
\begin{remark}[Local extractors]
Composing this with an extractor over $\mathcal{X}^{\ell}$ one obtains a \emph{local extractor}, which 
reads only a small fraction (specifically $\frac{\ell}{t}$) of input bits. We refer to~\cite{DBLP:conf/crypto/Vadhan03}
for a general discussion.
\end{remark}

\begin{proof}
We will argue that the sequence $X_{i_1},\ldots,X_{i_{\ell}}$, for $\ell$ sufficiently big, likely
has the same entropy rate (entropy per block) as the original sequence $X_1,\ldots,X_{\ell}$.

Let $\mathcal{B}$ be the family guaranteed by \Cref{thm:chain_rule}. By part (a) of \Cref{cor:chain_rule} applied $t$ times
(starting from $i=t$ downto $i=1$) we have
for every $B\in\mathcal{B}$
\begin{align*}
 \sum_{i=1}^{t}{H}_{\infty}(X_i|X_{<i},B) \geqslant H_{\infty}(X|B) - O(t\delta)
\end{align*}
and now by part (b) applied to each summand
\begin{align}\label{eq:app_1}
 \sum_{i=1}^{t}\widetilde{H}_{\infty}(X_i|X_{<i},B) \geqslant H_{\infty}(X|B) - O(t\delta).
\end{align}
Note that
\begin{align*}
 \mathbb{E}_{i_1,\ldots,i_{\ell}}\sum_{j=1}^{\ell}\widetilde{H}_{\infty}(X_{i_j}|X_{<i_{j}},B) = 
 \frac{1}{t}\sum_{i=1}^{t}\widetilde{H}_{\infty}(X_i|X_{<i},B)
\end{align*}
In particular, with high probability over $(i_j)_{j=1,\ldots,\ell}$ 
\begin{align}\label{eq:app_2}
\frac{1}{\ell} \sum_{j=1}^{\ell}\widetilde{H}_{\infty}(X_{i_j}|X_{<i_{j}},B) \gtrsim \frac{1}{t}\sum_{i=1}^{t}H_{\infty}(X_i|X_{<i},B).
\end{align}
For the sake of clarity, we comment later on the exact error in accuracy and probability in \Cref{eq:app_2}.
Observe that by part (b) of \Cref{cor:chain_rule} we obtain
\begin{multline}\label{eq:app_3}
\frac{1}{\ell} \sum_{j=1}^{\ell}\widetilde{H}_{\infty}(X_{i_j}|X_{i_{j}}X_{i_{j-1}}\ldots,X_{i_1}|B)  \geqslant \\  \frac{1}{\ell} \sum_{j=1}^{\ell}\widetilde{H}_{\infty}(X_{i_j}|X_{<i_{j}}|B) 
 -  O(\delta)
\end{multline}
(which is the \emph{conditioning reduces entropy} property). Again, by applying part (b) of \Cref{cor:chain_rule} to the sum on the right-hand side of \Cref{eq:app_3} we get
\begin{multline}\label{eq:app_4}
\frac{1}{\ell} \sum_{j=1}^{\ell}{H}_{\infty}(X_{i_j}|X_{i_{j}}X_{i_{j-1}}\ldots,X_{i_1}|B)  \geqslant \\
\frac{1}{\ell} \sum_{j=1}^{\ell}\widetilde{H}_{\infty}(X_{i_j}|X_{i_{j}}X_{i_{j-1}}\ldots,X_{i_1}|B)  - O(\delta)
\end{multline}
By \Cref{lemma:chaining_entropy} from \Cref{eq:app_3} we get for every $B'\in\mathcal{B}'$
\begin{multline}\label{eq:app_5}
\frac{1}{\ell}{H}_{\infty}(X_{i_{\ell}}X_{i_{\ell-1}}X_{i_{\ell-2}}\ldots,X_{i_1}|B) \geqslant \\
\frac{1}{\ell} \sum_{j=1}^{\ell}\widetilde{H}_{\infty}(X_{i_j}|X_{i_{j}}X_{i_{j-1}}\ldots,X_{i_1}|B)  - O(\delta).
\end{multline}
Combining this with \Cref{eq:app_3},\Cref{eq:app_2} and \Cref{eq:app_1} we finally obtain (with high probability)
\begin{multline}\label{eq:app_6}
\frac{1}{\ell} {H}_{\infty}(X_{i_{\ell}}X_{i_{\ell-1}}X_{i_{\ell-2}}\ldots,X_{i_1}|B) \gtrsim \\ \frac{1}{t}H_{\infty}(X|B) - O(\delta).
\end{multline}
Note that this holds for every $B$.
Recall that by \Cref{rem:fat_partition} we can assume $\Pr[B] = \Omega(|\mathcal{B}|^{-1}\epsilon)$. Now by \Cref{lemma:leakage}
we have $H_{\infty}(X|B) > H_{\infty}(X)-O(\log|\mathcal{B}| + \log(1/\epsilon))$
and thus
\begin{multline*}
\frac{1}{\ell} {H}_{\infty}(X_{i_{\ell}}X_{i_{\ell-1}}X_{i_{\ell-2}}\ldots,X_{i_1}|B) \gtrsim \\
\frac{1}{t}H_{\infty}(X|B) -\frac{O(\log|\mathcal{B}| + \log(1/\epsilon))}{t}.
\end{multline*}
We can do slightly better. Namely, from \Cref{eq:app_6} 
\begin{multline}\label{eq:app_7}
\frac{1}{\ell}\widetilde{{H}}_{\infty}(X_{i_{\ell}}X_{i_{\ell-1}}X_{i_{\ell-2}}\ldots,X_{i_1}|\mathcal{B}) \gtrsim \\
\frac{1}{t}H_{\infty}(X|\mathcal{B})  - O(\delta).
\end{multline}
where $\mathcal{B}$ is a random variable that assigns to every point $x$ the corresponding set $B$ covering $x$, conditioned
in addition on the map being defined (it fails when $x$ is not covered by any $B$ which happens w.p. at most $\epsilon$)
Indeed, we have
\begin{multline*}
2^{-H_{\infty}(X_{i_{\ell}}X_{i_{\ell-1}}X_{i_{\ell-2}}\ldots,X_{i_1}|B)} \leqslant \\
 \left( 2^{-H_{\infty}(X_{i_{t}}X_{i_{t-1}}X_{i_{\ell-2}}\ldots,X_{i_1})} \right)^{\frac{\ell}{t}}.
\end{multline*}
and \Cref{eq:app_7} follows by averaging over $B$ and the Jensen Inequality (note that $\ell < t$ implies that the corresponding mapping is concave).
By \Cref{lemma:chaining_entropy} applied to $X|\mathcal{B}=B$ for all possible $B$ to outcomes of $\mathcal{B}$ we obtain
\begin{multline}\label{eq:app_7}
\frac{1}{\ell}\widetilde{{H}}_{\infty}(X_{i_{\ell}}X_{i_{\ell-1}}X_{i_{\ell-2}}\ldots,X_{i_1}|\mathcal{B}) \gtrsim \\
\frac{1}{t}H_{\infty}(X) -\frac{\log |\mathcal{B}|}{t} - O(\delta)
\end{multline}
It remains to use an explicit bound on $|\mathcal{B}|$ from \Cref{thm:chain_rule}, set the sampler to $\gamma = \epsilon$ and compute $\theta$ from $\epsilon$ and $\ell$.
\end{proof}

\section{Conclusion}\label{sec:conclusion} 
By a simple combinatorial argument combined with the spoiling knowledge technique we showed how to exhibits strong block-source structures
in any min-entropy source. This approach may be applied to locally enforce chain rules (or other desired properties) for min-entropy.






%

\printbibliography

\end{document}